\documentclass[envcountsect,11pt]{llncs}

\usepackage{graphicx} 
\usepackage{float} 
\usepackage[utf8]{inputenc}
\usepackage{amsmath,algorithm,algorithmic}
\usepackage{amssymb,enumerate}
 
\usepackage{amsthm}
\usepackage{fullpage}
 
\usepackage{thmtools,thm-restate}
\usepackage{setspace}
\usepackage{xspace}
\usepackage{color}
\usepackage{pdfsync}
\usepackage{cleveref}
\usepackage{tikz}
\usepackage{subcaption}
\newcommand{\setof}[2]{\{#1\mid#2\}}

\newcommand{\apdx}{Appendix\xspace}

\newcommand{\bigraphs}{{\mathcal B}_{n\times n}}
\newcommand{\digraphs}{{\mathcal D}_n}
\newcommand{\len}{\textit{len}}
\newcommand{\pre}{\text{pre}}
\newcommand{\suf}{\text{suf}}
\newcommand {\ov}{\text{ov}}

\newcommand{\ind}{hub}

\newcommand{\Index}{HUB\xspace}

\newcommand{\readability}{readability\xspace}
\newcommand{\readabilities}{readabilities\xspace}
\newcommand{\readabilityFull}{readability\xspace}
\newcommand{\ReadabilityFull}{Readability\xspace}
\newcommand{\readabilityParameter}{r}

\begin{document}

\title{On the Readability of Overlap Digraphs\thanks{This is a full
version of a conference paper of the same title at the
26th Annual Symposium on Combinatorial Pattern Matching (CPM 2015)}}
\author{Rayan Chikhi\inst{1,2}\and
Paul Medvedev\inst{2}
\and Martin Milani\v{c}\inst{3}
\and Sofya Raskhodnikova\inst{2}}

\institute{$^1$CNRS, UMR 9189, France, $^2$The Pennsylvania State University, USA, $^3$University of Primorska, Slovenia}
\maketitle
\begin{abstract}
We introduce the graph parameter {\em \readabilityFull} and study it as a function of the number of vertices in a graph.
	Given a digraph $D$, an injective overlap labeling assigns a unique string to each vertex such that
	there is an arc from $x$ to $y$  if and only if $x$ properly overlaps $y$.
    The \readabilityFull of $D$ is the minimum string length for which an injective overlap labeling exists.
	In applications that utilize overlap digraphs (e.g., in bioinformatics),
	\readabilityFull reflects the length of the strings from which the overlap digraph is constructed.
	We study the asymptotic behaviour of \readabilityFull by casting it in purely graph theoretic terms (without any reference to strings). We
	prove upper and lower bounds on \readabilityFull for certain graph families and general graphs.
\end{abstract}

\section{Introduction}

In this paper, we introduce and study a graph parameter called \readability, motivated by applications
of overlap graphs in bioinformatics.
A string $x$ {\em overlaps} a string $y$ if there is a suffix of $x$
that is equal to a prefix of $y$. They overlap {\em properly} if, in addition, the suffix and prefix are both proper.
The {\em overlap digraph} of a set of strings $S$ is a digraph where each string is a vertex and
there is an arc from $x$ to $y$ (possibly with $x=y$) if and only if $x$ properly overlaps $y$.
Walks in the overlap digraph of $S$ represent
strings that can be spelled by stitching strings of $S$ together,
using the overlaps between them.
Overlap digraphs have various applications, e.g., they are used by approximation algorithms for the Shortest Superstring Problem \cite{sweedyk00}.
Their most impactful application, however, has been in bioinformatics.
Their variants, such as de Bruijn graphs \cite{IW95} and string graphs \cite{M05},
have formed the basis of nearly all genome assemblers used today (see \cite{MKS10,NP13} for a survey),
successful despite results showing that assembly is a hard problem in theory~\cite{BBT13,NP09,MGMB07}.
In this context,
the strings of $S$ represent known fragments of the genome (called {\em reads}),
and the genome is represented by walks in the overlap digraph of $S$.
However, do the overlap digraphs generated in this way capture all possible digraphs, or
do they have any properties or structure that can be exploited?

Braga and Meidanis \cite{BM02} showed that overlap digraphs capture all possible digraphs,
i.e., for every digraph $D$, 
there exists a set of strings $S$ such that their overlap digraph is $D$.
Their proof takes an arbitrary digraph and shows how to construct an {\em injective overlap labeling},
that is, a function assigning a unique string to each vertex,
such that $(x,y)$ is an arc if and only if the string assigned to $x$ properly
overlaps the string assigned to $y$.
However, the {\em length} of strings
produced by their method can be exponential in the number of vertices.
In the bioinformatics context, this is unrealistic,
as the read size is typically much smaller than the number of reads.

To investigate the relationship between the string length and the number of vertices, we introduce a graph parameter called {\em \readability}.
The \readability of a digraph $D$, denoted $\readabilityParameter(D)$,
is the smallest nonnegative integer $r$ such that there exists an injective overlap labeling of $D$
with strings of length $r$.
The result by \cite{BM02} shows that \readability is well defined and is at most $2^{\Delta+1}-1$,
where $\Delta$ 
is the maximum of the in- and out-degrees of vertices in $D$.
However, nothing else is known about the parameter,
though there are papers that look at related notions~\cite{BFKSW02,BFKK02,BHKW99,GP14,LZ07,LZ10,PSW03,TU88}.

In this paper, we study the asymptotic behaviour of \readability as a function of the number of vertices in a graph. We define \readability for  undirected bipartite graphs and show that
the two definitions of \readability are asymptotically equivalent.
We capture \readability using purely graph theoretic parameters (i.e., without any reference to strings).
For trees, we give a parameter that characterizes \readability exactly.
For the larger family of bipartite $C_4$-free graphs, we give a parameter
that approximates \readability to within a factor of $2$.
Finally, for general bipartite graphs, we give a parameter that
is bounded on the same sets of graphs as \readability.

We apply our purely graph theoretic interpretation to
prove \readability upper and lower bounds on several graph families.
We show, using a counting argument,
that almost all digraphs and bipartite graphs have \readability of at least $\Omega(n/\log n)$.
Next, we construct a graph family inspired by Hadamard codes and prove that it has \readability $\Omega(n)$.
Finally, we show that the \readability of trees is bounded from above by their radius,
and there exist trees of arbitrary \readability that achieve this bound.

\section{Preliminaries}
\subsubsection*{General definitions and notation.}\label{sec:dfn}
Let $x$ be a string.
We denote the length of $x$ by $|x|$. We use $x[i]$ to refer to the $i^{\text{th}}$ character of $x$, and
denote by $x[i..j]$ the substring of $x$ from the $i^{\text{th}}$ to the $j^{\text{th}}$ character, inclusive.
We let $\pre_i(x)$ denote the prefix $x[1..i]$ of $x$, and we let $\suf_i(x)$ denote the suffix $x \left[|x| - i + 1 .. |x|\right]$.
Let $y$ be another string.
We denote by $x\cdot y$ the concatenation of $x$ and $y$.
We say that $x$ {\em  overlaps} $y$ if there exists an $i$ with $1 \leq i \le \min\{|x|,|y|\}$
such that $\suf_i(x) = \pre_i(y)$.
In this case, we say that $x$ overlaps $y$ by $i$.
If $i < \min\{ |x|, |y| \}$, then we call the overlap {\em proper}.
Define $\ov(x,y)$ as the minimum $i$ such that $x$ overlaps $y$ by $i$, or $0$ if $x$ does not overlap $y$.
For a positive integer $n$, we denote by $[n]$ the set $\{1,\ldots, n\}$.

We refer to finite simple undirected graphs simply as graphs and
to finite directed graphs without parallel arcs in the same direction as digraphs.
For a vertex $v$ in a graph, we denote the set of neighbors of $v$ by $N(v)$.
A {\em biclique} is a complete bipartite graph.
Note that the one-vertex graph is a biclique (with one of the parts of its bipartition being empty).
Two vertices $u,v$ in a graph are called {\em twins} if they have the same neighbors, i.e., if $N(u)=N(v)$.
If, in addition, $N(u)=N(v)\neq\emptyset$, vertices $u,v$ are called {\em non-isolated twins}.
A {\em  matching} is a graph of maximum degree at most $1$,
though we will sometimes slightly abuse the terminology and not distinguish between matchings and their edge sets.
A cycle (respectively, path) on $i$ vertices is denoted by $C_i$ (respectively, $P_i$).
For graph terms not defined here, see, e.g.,~\cite{BM08}.

\subsubsection*{\ReadabilityFull~of digraphs.}
A {\em  labeling} $\ell$ of a graph or digraph is a function assigning a string to each vertex
such that all strings have the same length, denoted by $\len(\ell)$.
We define $\ov_\ell(u,v) = \ov(\ell(u), \ell(v))$.
An {\em  overlap labeling} of a digraph $D=(V,A)$ is a labeling $\ell$ such that
$(u,v) \in A$ if and only if $0 < \ov_\ell(u,v)) < \len(\ell)$.
An overlap labeling is said to be {\em  injective} if it does not generate duplicate strings.
Recall that the \readability of a digraph $D$, denoted $\readabilityParameter(D)$,
is the smallest nonnegative integer $r$ such that there exists an injective overlap labeling of $D$ of length $r$.
We note that in our definition of readability we do not place any restrictions on the alphabet size.
Braga and Meidanis \cite{BM02} gave a reduction from an overlap labeling of length $\ell$ over an arbitrary alphabet $\Sigma$
to an overlap labeling  of length $\ell \log|\Sigma|$ over the binary alphabet.

\subsubsection*{\ReadabilityFull~of bipartite graphs.}
We also define a modified notion of \readability that
applies to balanced bipartite graphs as opposed to digraphs.
We found that \readability on balanced bipartite graphs is simpler to study but is asymptotically equivalent
to \readability on digraphs.
Let $G=(V,E)$ be a bipartite graph with a given bipartition of its vertex set $V(G) = V_s\cup V_p$.
(We also use the notation $G = (V_s,V_p,E)$.)
We say that $G$ is {\em balanced} if $|V_s| = |V_p|$.
An {\em  overlap labeling of $G$} is a labeling $\ell$ of $G$ such that
for all $u \in V_s$ and $v \in V_p$,
$(u,v) \in E$ if and only if $\ov_\ell(u,v) > 0$.
In other words, overlaps are exclusively between the suffix of a string assigned to a vertex in $V_s$ and the prefix of a string assigned to a vertex in $V_p$.
The {\em  \readability} of $G$ is the smallest nonnegative integer $r$
such that there exists an overlap labeling of $G$ of length $r$.
Note that we do not require injectivity of the labeling,
nor do we require the overlaps to be proper.
As before, we use $\readabilityParameter(G)$ to denote the \readability of $G$.

We note that in our definition of readability we do not place any restrictions on the alphabet size.
Braga and Meidanis \cite{BM02} gave a reduction from an overlap labeling of length $\ell$ over an arbitrary alphabet $\Sigma$
to an overlap labeling  of length $\ell \log|\Sigma|$ over the binary alphabet.

For a labeling $\ell$,
we define ${\it inner}_i(\ell(v)) = \suf_i(\ell(v))$ if $v \in V_s$ and
${\it inner}_i(\ell(v)) = \pre_i(\ell(v))$ if $v \in V_p$.
Similarly, we define ${\it outer}_i(\ell(v)) = \pre_i(\ell(v))$ if $v \in V_s$ and
${\it outer}_i(\ell(v)) = \suf_i(\ell(v))$ if $v \in V_p$.

Let $\bigraphs$ be the set of balanced bipartite graphs with nodes $[n]$ in each part, and
let $\digraphs$ be the set of all digraphs with nodes $[n]$.
The \readabilities~of digraphs and of bipartite graphs are connected by the following theorem,
which implies that they are asymptotically equivalent.

\begin{restatable}{theorem}{restatabletransformation}\label{thm:transformation}
There exists a bijection $\psi:\bigraphs\to\digraphs$
with the property
that for any $G\in\bigraphs$ and $D\in\digraphs$,
such that $D = \psi(G)$,
we have that $\readabilityParameter(G) < \readabilityParameter(D) \leq 2 \cdot \readabilityParameter(G) + 1$.
\end{restatable}
As a result, we can study \readability of balanced bipartite graphs,
without asymptotically affecting our bounds.
For example, 
we show in Section~\ref{sec:hadamard-graphs} (in~\Cref{thm:hadamard})
that there exists a family of balanced bipartite graphs with \readability $\Omega(n)$,
which leads to the existence of digraphs with \readability $\Omega(n)$.

\section{Graph theoretic characterizations}
In this section, we relate \readability of balanced bipartite graphs to several purely graph theoretic parameters,
without reference to strings.

\subsection{Trees and $C_4$-free graphs}\label{sec:trees_and_C4-free}

For trees, we give an exact characterization of \readability,
while for $C_4$-free graphs, we give a parameter that is a $2$-approximation to \readability.
A {\em decomposition of size $k$} of a bipartite graph $G=(V_s, V_p, E)$ is
a function on the edges of the form
$w:E\rightarrow [k]$.
Note that a labeling $\ell$ of $G$ implies a decomposition of $G$, defined by
$w(e) = \ov_\ell(e)$ for all $e\in E$.
We call this the {$\ell$-decomposition}.
We say that a labeling $\ell$ of $G$ {\em achieves} $w$ if it is an overlap labeling and $w$ is the $\ell$-decomposition.
Note that we can express \readability as
$$\readabilityParameter(G) = \min \{k \mid \text{$w$ is a decomposition of size $k$}\,, 
\text{$\exists$ a labeling $\ell$ that achieves $w$\}}\,.$$
Our goal is to characterize in graph theoretic terms the properties of $w$ which are satisfied if and only if
$w$ is the $\ell$-decomposition, for some $\ell$.
While this proves challenging in general, we can achieve this for trees using a
condition which we call the $P_4$-rule.
We say that $w$ satisfies the {\em $P_4$-rule} if for every induced four-vertex path
$P=(e_1, e_2, e_3)$ in $G$, the following condition holds:
if $w(e_2) = \max\{w(e_1), w(e_2), w(e_3)\}$, then
$w(e_2) \geq w(e_1) + w(e_3)$.
We will prove:
\begin{sloppypar}
\begin{restatable}{theorem}{restatabletree}\label{thm:tree}
	Let $T$ be a tree.
	Then $\readabilityParameter(T) = \min \{k \mid w$ \text{is a decomposition of size $k$ that} \text{satisfies the $P_4$-rule}$\}$.
\end{restatable}
\end{sloppypar}
Note that for cycles, the equality does not hold.
For example, consider the decomposition $w$ of $C_6$
given by the weights $2, 4, 2, 2, 3, 1$.
This decomposition satisfies the $P_4$ rule but it can be shown using case analysis
that there does not exist a labeling $\ell$ achieving $w$.

However, we can give a characterization of \readability for $C_4$-free graphs in terms of a parameter that is asymptotically equivalent to \readability,
using a condition which we call the strict $P_4$-rule.
The strict $P_4$-rule is identical to the $P_4$-rule accept that the inequality becomes strict.
That is, $w$ satisfies the {\em strict $P_4$-rule} if for every induced four-vertex path
$P=(e_1, e_2, e_3)$,
if $w(e_2) = \max\{w(e_1), w(e_2), w(e_3)\}$, then
$w(e_2) > w(e_1) + w(e_3)$.
Note that a decomposition that satisfies the strict $P_4$-rule automatically satisfies the $P_4$-rule,
but not vice-versa.
We will prove:
\begin{sloppypar}
	\begin{restatable}{theorem}{restatablecfour}\label{thm:cfour}
	Let $G$ be a $C_4$-free bipartite graph.
	Let $t = \min \{k \mid w$ \text{is a decomposition of } \text{size $k$ that satisfies the strict $P_4$-rule$\}$}.
	Then $t / 2 < \readabilityParameter(G) \leq t$.
\end{restatable}
\end{sloppypar}
We note that this characterization cannot be extended to graphs with a $C_4$.
The example in \Cref{fig:cfour} shows a graph with a decomposition which satisfies the strict $P_4$-rule but
it can be shown using case analysis that there does not exists a labeling $\ell$ achieving this decomposition.
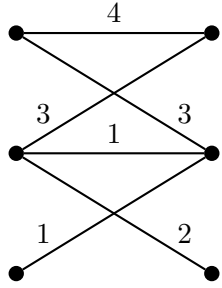
\begin{figure}[t]
  \centering
        \begin{tikzpicture}[-,shorten >=0pt,auto,node distance=2.6cm,
            thick]
            \tikzstyle{main node}=[circle, draw, fill=black, inner sep=1pt, minimum width=5pt]
            \node[main node] (1) {};
            \node[main node] (2) [right of=1] {};
            \node[main node] (3) [below of=1, yshift=1cm] {};
            \node[main node] (4) [right of=3] {};
            \node[main node] (5) [below of=3, yshift=1cm] {};
            \node[main node] (6) [right of=5] {};

            \path[]
            (1) edge node [above] {4} (2)
            (3) edge node [below right, xshift=.7cm] {3} (2)
            (3) edge node [above] {1} (4)
            (3) edge node [below right, xshift=.7cm] {2} (6)
            (5) edge node [below left, xshift=-.7cm] {1} (4)
            (1) edge node [below left, xshift=-.7cm] {3} (4);
        \end{tikzpicture}
    \caption{Illustration that \Cref{thm:cfour} cannot be extended to graphs with a $C_4$.
    Example of a graph and decomposition that satisfies the strict $P_4$-rule,
yet no overlap labeling $\ell$ exists that achieves it.} 
        \label{fig:cfour}
\end{figure}

In the remainder of this section, we will prove these two theorems.
We first show that an $\ell$-decomposition satisfies the $P_4$-rule (proof in the \apdx).
\begin{restatable}{lemma}{restatablefourstrings}\label{lem:fourstrings}
Let $\ell$ be an overlap labeling of a bipartite graph $G$.
Then the $\ell$-decomposition satisfies the $P_4$-rule.
\end{restatable}
Now, consider a $C_4$-free bipartite graph $G = (V_s, V_p, E)$ and let $w$ be a decomposition satisfying the $P_4$-rule.
We will prove both \Cref{thm:tree} and \Cref{thm:cfour} by constructing the following labeling.
Let us order the edges $e_1, \ldots, e_{|E|}$ in order of non-decreasing weight.
For $0\le j\le |E|$, we define the graph $G^j = (V_s, V_p, \setof{e_i \in E}{ i \leq j})$.
For a vertex $u$, define ${\it len}_j(u) = \max \setof{w(e_i)}{i \leq j, e_i \text{ is incident with $u$}}$, if
the degree of $u$ in $G^j$ is positive, and $0$ otherwise.
We will recursively define a labeling $\ell_j$ of $G^j$
such that $|\ell_j(u)| = {\it len}_j(u)$ for all $u$.
The initial labeling $\ell_0$ assigns $\epsilon$ to every vertex.
Suppose we have a labeling $\ell_j$ for $G^j$, and
$e_{j+1} = (u,v)$.
Recall that because $w$ satisfies the $P_4$-rule and $G$ is $C_4$-free,
$w(u,v) \geq {\it len}_j(u) + {\it len}_j(v) = |\ell_j(u)| + |\ell_j(v)|$.
(Note that the inequality holds also in the case when one of the two summands is $0$.)
Let $A$ be a (possibly empty) string of length $w(u,v) - |\ell_j(u)| - |\ell_j(v)|$
composed of non-repeating characters that do not exist in $\ell_j$.
Define $\ell_{j+1}$ as $\ell_{j+1}(x) = \ell_j(x)$ for all $x \notin \{u,v\}$, and
	$\ell_{j+1} (u)  = \ell_{j+1}(v) = \ell_j (v) \cdot A \cdot \ell_j (u)$.
We denote the labeling of $G$ as $\ell = \ell_{|E|}$.
We will slightly abuse notation in this section,
ignoring the fact that a labeling must have labels of the same length.
This is inconsequential, because strings can always be padded from the beginning or end with
distinct characters without affecting any overlaps.

First, we state a useful Lemma, that two vertices share a character in the labeling only if they are connected
by a path (proof in the \apdx).
\begin{restatable}{lemma}{restatablepathexists}
	\label{lem:pathexists}
	Let $c$ be a character that is contained in $\ell_j(u)$ and in $\ell_j(v)$, for some pair of distinct vertices.
	Then there exists a path between $u$ and $v$ in $G^j$.
\end{restatable}
We are now ready to show that $\ell$ achieves $w$ for trees, and,  if $w$ also satisfies the strict $P_4$-rule,
for $C_4$-free graphs.
\begin{restatable}{lemma}{restateachieve}
	\label{lem:achieve}
	Let $G$ be a $C_4$-free bipartite graph and let $w$ be a decomposition that satisfies the $P_4$-rule.
	Then the above defined labeling $\ell$ achieves $w$ if $w$ satisfies the strict $P_4$-rule or if $G$ is acyclic.
\end{restatable}
\begin{proof}
	We prove by induction on $j$ that $\ell_j$ achieves $w$ on $G^j$.
	Suppose that the Lemma holds for $\ell_j$ and consider the effect of adding $e_{j+1} = (u,v)$.
	Notice that to obtain $\ell_{j+1}$ we only change labels by adding outer characters, hence,
	any two vertices that overlap by $i$ in $\ell_j$ will also overlap by $i$ in $\ell_{j+1}$.
	Moreover, only the labels of $u$ and $v$ are changed, and
	an overlap between $u$ and $v$ of length $w(u,v)$ is created.
	It remains to show that no shorter overlap is created between $u$ and $v$ and
	that no new overlap is created involving $u$ or $v$, except the one between $u$ and $v$.

	First, consider the case when $w(u,v) > |\ell_j(u)| + |\ell_j(v)|$ and so the middle string ($A$) of the new labels is non-empty.
	Because the characters of $A$ do not appear in $\ell_j$, we do not create any new overlaps except besides the one between
    $u$ and $v$ and the only overlap between $u$ and $v$ must be of length $w(u,v)$ since the characters of $A$ must align.
	Thus $\ell_{j+1}$ achieves $w$ on $G^{j+1}$.

	Next, consider the case when
	$w(u,v) = |\ell_j(v)|$ (the case when $w(u,v) = |\ell_j(u)|$ is symmetric).
	In this case, $A = \epsilon$, $\ell_j(u) = \epsilon$, and $|\ell_j(v)|>0$ (since $w(u,v)>0$).
	Suppose for the sake of contradiction that there exists a vertex $v' \ne v$ such that
	$(u,v')$ is not an edge but ${\it inner}_k (\ell_{j+1}(u)) = {\it inner}_k (\ell_{j+1}(v'))$,
	for some $0 < k \leq w(u,v)$.
	We know, from the construction of $\ell_j$, that there exists a vertex $u'$ such that $w(u',v) = |\ell_j (v)|$.
	We then have ${\it inner}_k(\ell_j(u')) = {\it outer}_k(\ell_j(v)) = {\it inner}_k(\ell_{j+1}(u)) = {\it inner}_k(\ell_{j+1}(v')) = {\it inner}_k (\ell_j(v'))$.
	By the induction hypothesis, there is an edge $(u',v')$ and $w(u',v') \leq k$. 
	The edges $(u,v), (v, u'), (u', v')$ form a $P_4$, which is also induced because $G$ is $C_4$-free.
	Because $w(u,v) = w(u',v) \geq w(u',v')>0$, the $P_4$-rule is violated, a contradiction.
	Therefore no new overlaps are created involving $u$.
	To show that there are no overlaps from $u$ to $v$ smaller than $w(u,v)$,
	observe that any such overlap would also be an overlap between $u'$ and $v$ that is smaller than $w(u',v)$,
	contradicting the induction hypothesis.
	Therefore, $\ell_{j+1}$ achieves $w$ on $G^{j+1}$.

	It remains to consider the case when
	$w(u,v) = |\ell_j(u)| + |\ell_j(v)|$
	and $\ell_j(u) \ne \epsilon \ne \ell_j(v)$.
	We first show that this case cannot arise if $w$ satisfies the strict $P_4$-rule.
	There must exist edges in $G^j$ of weights $|\ell_j(u)|$ and $|\ell_j(v)|$ incident with  $u$ and $v$, respectively.
	These edges, together with $(u,v)$ in the middle, form a $P_4$,
	which must be induced since $G$ does not contain a $C_4$.
	Furthermore, $(u,v)$ achieves the maximum weight.
	The strict $P_4$-rule implies $w(u,v) > |\ell_j(u)| + |\ell_j(v)|$, a contradiction.
	
	Now, assume that $G$ is acyclic, and suppose for the sake of contradiction that the new labeling creates an overlap between $v$ and a
    vertex $u' \ne u$ (the case of an overlap between $u$ and $v' \ne v$ is symmetric).
	Consider the character $c$ at position $|\ell_j(v)| + 1$ of $\ell_{j+1}(v)$.
	The length of the overlap between $\ell_{j+1}(v)$ and $\ell_{j+1}(u')=\ell_{j}(u')$ must be greater than $|\ell_j(v)|$,
    otherwise it would have been an overlap in $\ell_j$.	Thus, $\ell_j(u')$ must contain $c$.
	By construction of $v$'s new label, $\ell_{j} (u)$ must also contain $c$.
	Applying \Cref{lem:pathexists}, there must be a path between $u'$ and $u$ in $G^j$.
	On the other hand, the overlap between $v$ and $u'$ spans $(\ell_{j}(v))[1]$, and hence $\ell_{j}(v)$ and $\ell_{j}(u')$ must share a character.
	Applying \Cref{lem:pathexists}, 
	there must exist a path between $u'$ and $v$ in $G^j$.
	Consequently, there exists a path from $u$ to $v$ in $G^j$.
	Combining this path with $e_{j+1} = (u,v)$, we get a cycle in $G^{j+1}$, which is a contradiction.

	Finally suppose, for the sake of contradiction, that $\ell_{j+1}(u)$ overlaps $\ell_{j+1}(v)$ by some $k < w(u,v)$.
	By the induction hypothesis, $k > |\ell_j(v)|$.
	Consider the last character $c$ of $\ell_j(v)$.
	It must also appear as the inner position $i = k - |\ell_j(v)| + 1$ in $\ell_{j+1}(u)$.
	Since $k \le  w(u,v)-1$, we have $i \leq w(u,v) - |\ell_j(v)| = |\ell_j(u)|$, and
	the $i^{\text{th}}$ inner position in $\ell_{j+1}(u)$ is also the
	the $i^{\text{th}}$ inner position in $\ell_j(u)$.
	Applying \Cref{lem:pathexists} to $c$ in $\ell_j(v)$ and $\ell_j(u)$, there must exist a path between $u$ and $v$ in $G^j$.
	Combining this path with $e_{j+1} = (u,v)$, we get a cycle in $G^{j+1}$, which is a contradiction.
\end{proof}
We can now prove Theorems~\ref{thm:tree} and~\ref{thm:cfour}.
\begin{proof}[Proof of~\Cref{thm:tree}]
	Let $t = \min \{k \mid \text{$w$ is a decomposition of size $k$ that satisfies the $P_4$-rule}\}$.
	First, let $w$ be a decomposition of size $t$ satisfying the $P_4$-rule.
	\Cref{lem:achieve} states that the above defined labeling $\ell$ achieves $w$ and so $\readabilityParameter(T) \leq \max_e(w_e) = t$.
	For the other direction,
	consider an overlap labeling $b$ of $T$ of minimum length.
	By \Cref{lem:fourstrings}, the $b$-decomposition satisfies the $P_4$-rule.
	Hence,
	$\readabilityParameter(T) = \len(b) \geq t$. 
\end{proof}
\begin{proof}[Proof of~\Cref{thm:cfour}]
	Let $w$ be a decomposition of size $t$ satisfying the strict $P_4$-rule.
	By \Cref{lem:achieve}, the above defined labeling $\ell$ achieves $w$ and so $\readabilityParameter(G) \leq \max_e(w_e) = t$.
	On the other hand, let $b$ be an overlap labeling of length $\readabilityParameter(G)$.
	Define $w(e) = 2\ov_b(e) - 1$, for all $e \in E(G)$.
	We claim that $w$ satisfies the strict $P_4$-rule, which will imply that $t \leq \max_e w(e) = 2\readabilityParameter(G) - 1$.
	To see this, let $e_1, e_2, e_3$ be the edges of an arbitrary induced $P_4$.
	Observe that $w(e_2) = \max \{w(e_1), w(e_2), w(e_3)\}$ if and only if
	$\ov_b(e_2) = \max \{\ov_b(e_1), \ov_b(e_2), \ov_b(e_3)\}$.
	Furthermore, it can be algebraicly verified that if
	$\ov_b(e_2) \geq \ov_b(e_1) + \ov_b(e_3)$ then
	$w(e_2) > w(e_1) + w(e_3)$.
	By \Cref{lem:fourstrings}, the $b$-decomposition satisfies the $P_4$-rule and,
	therefore, $w$ satisfies the strict $P_4$-rule.
\end{proof}

\subsection{General graphs}\label{sec:hub-up}
In the previous subsection, we derived graph theoretic characterizations of \readability that are exact for trees and approximate for $C_4$-free bipartite graphs.
Unfortunately, for a general graph, it is not clear how to construct an overlap labeling
from a decomposition satisfying the $P_4$-rule (as we did in \Cref{lem:achieve}).
In this subsection, we will consider an alternate rule (HUB-rule),
which we then use to construct an overlap labeling.

Given $G = (V_s, V_p, E)$ and a decomposition $w$ of size $k$,
we define $G^w_i$, for $i\in[k]$, as a graph with the same vertices as $G$ and edges given by
$E(G^w_i) = \setof{e \in E}{w(e) = i}$.
When $w$ is obvious from the context, we will write $G_i$ instead of $G^w_i$.
Observe that the edge sets of $G^w_1,\dots, G^w_k$ form a partition of $E$.
We say that $w$ satisfies the {\em hierarchical-union-of-bicliques rule}, abbreviated as the {\em HUB-rule},
if the following conditions hold:
i) for all $i\in[k]$, $G^w_i$ is a disjoint union of bicliques, and
ii) if two distinct vertices $u$ and $v$ are non-isolated twins in $G^w_i$ for some
$i\in\{2,\dots,k\}$ then, for all $j\in [i-1]$,
$u$ and $v$ are (possibly isolated) twins  in $G^w_j$.
An example of a decomposition satisfying the HUB-rule is any $w:E\to[k]$
such that $G^w_1$ is an (arbitrary) disjoint union of bicliques and $G^w_2,\dots,G^w_k$ are matchings.
We can show that the decomposition implied by any overlap labeling must satisfy the HUB-rule
(proof in the \apdx).
\begin{restatable}{lemma}{restatablehubrule}\label{lem:hubrule}
Let $\ell$ be an overlap labeling of a bipartite graph $G$.
Then the $\ell$-decomposition satisfies the HUB-rule.
\label{prop:ldecomp}
\end{restatable}

We define {\em the \Index number of $G$} as
the minimum size of a decomposition of $G$  that satisfies the HUB-rule, and denote it by $\ind(G)$.
Observe that a decomposition of a graph into matchings (i.e. each $G_i^w$ is a matching) satisfies the HUB-rule.
By K\"onig's Line Coloring Theorem, any bipartite graph $G$ can be decomposed into $\Delta(G)$ matchings,
where $\Delta(G)$ is the maximum degree of $G$.
Thus, $\ind(G) \in [\Delta(G)]$.
Clearly, a graph $G$ has $\ind(G) = 1$ if and only if $G$ is a disjoint union of bicliques.
The \Index number captures \readability in the sense that the \readability of a graph family is bounded
(by a uniform constant independent of the number of vertices)
if and only if its \Index number is bounded.
This is captured by the following theorem:
\begin{theorem}\label{thm:characterization}
Let $G$ be a bipartite graph. Then $\ind(G)\leq \readabilityParameter(G)\leq 2^{\ind(G)}-1$.
\end{theorem}
In the remainder of this section, we will prove this theorem.
The first inequality directly follows from Lemma~\ref{lem:hubrule} because, by definition of \readability, there exists an overlap labeling $\ell$ of length $\readabilityParameter(G)$. Then the $\ell$-decomposition of $G$ is of size $\readabilityParameter(G)$ and satisfies the HUB-rule, implying $\ind(G)\leq \readabilityParameter(G)$.
To prove the second inequality, we will need to show:
\begin{restatable}{lemma}{restategeneralizedBM}\label{lem:generalized-BM}
Let $w$ be a decomposition of size $k$ satisfying the HUB-rule of a bipartite graph $G$.
Then there is an overlap labeling of $G$ of length $2^k-1$.
\end{restatable}
The second inequality of \Cref{thm:characterization}
follows directly by choosing a minimum decomposition
satisfying the HUB-rule, in which case $k=\ind(G)$.
Thus, it only remains to prove \Cref{lem:generalized-BM}.

We now define the labeling $t$ that is used to prove \Cref{lem:generalized-BM}.
Our construction of the labeling applies the following operation due to Braga and Meidanis~\cite{BM02}.
Given two vertices $u \in V_s$ and $v \in V_p$, a labeling $t$, and a filler character $a$ not used by $t$,
the {\it BM operation} transforms $t$
by relabeling both $u$ and $v$ with $t(v) \cdot a \cdot t(u)$.

We start by labeling $G_1$ as follows: each biclique $B$ in $G_1$ gets assigned a unique character $a_B$, and
each node $v$ in a biclique $B$ gets label $t(v)=a_B$.
Next, for $i\in[k-1],$ we iteratively construct a labeling of $G_1\cup\dots\cup G_{i+1}$ from a labeling $t$ of $G_1\cup\dots\cup G_i$.
We show by induction that the constructed labeling has an additional property that all twins in $G_1\cup\dots\cup G_{i+1}$
have the same labels and that the length of the labeling is $2^{i+1} - 1$.
Observe that the labeling of $G_1$ satisfies this property.

We choose a unique (not previously used) character $a_B$ for each biclique $B$ of $G_{i+1}$.
If $B$ consists of a single vertex $v$,
then we assign to $v$ the label $a_B\cdot t(v)$ if $v\in V_s$, and
$t(v)\cdot a_B$ if $v\in V_p$. Otherwise,
since $w$ satisfied the HUB-rule, all vertices in $B \cap V_s$ are twins in $G_1\cup\dots\cup G_i$ and,
by the induction hypothesis, are assigned the same labels in $t$.
Analogously, $t$ will assign the same labels to all nodes in $B \cap V_p$.
Consider an arbitrary edge $(u,v)$ in $B$.
We apply the BM operation with character $a_B$ to $(u,v)$ and assign the resulting label $t(v) \cdot a_B \cdot t(u)$ to all nodes in $B$.
This completes the construction of labeling of $G_1\cup\dots\cup G_{i+1}$. Observe that it assigns the same labels to all twins in $G_1\cup\dots\cup G_{i+1}$,
and that the length is $2^{i+1} - 1$.
To complete the proof of \Cref{thm:characterization},
we show in the \apdx that the final labeling is an overlap labeling of $G$.

Note that if $w$ is a decomposition into matchings,
then our labeling algorithm behaves identically to the Braga-Meidanis (BM) algorithm \cite{BM02}.
However, in the case that $w$ is of size $o(\Delta(G))$,
our labeling algorithm gives a better bound than BM.
For example, for the $n \times n$ biclique, our algorithm gives a labeling of length $1$,
while BM gives a labeling of length $2^n - 1$.

\section{Lower and upper bounds on \readability}
In this section, we prove several lower and upper bounds on \readability,
making use of the characterizations of the previous section.

\subsection{Almost all graphs have \readability $\Omega(n/\log n)$}\label{sec:counting-lb}

In this subsection, we show that, in both the bipartite and directed graph models,
there exist graphs with \readability at least
$\Omega(n/\log n)$,
and that in fact almost all graphs have at least this \readability.
\begin{restatable}{theorem}{restatablebinary}\label{lem:binary}
Almost all graphs in $\bigraphs$ (and, respectively, $\digraphs$) have \readability $\Omega(n/\log n)$.
When restricted to a constant sized alphabet,
almost all graphs in $\bigraphs$ (and, respectively, $\digraphs$) have \readability $\Omega(n)$.
\end{restatable}
\begin{proof}[Proof (constant sized alphabet case)]
We prove the lemma by a counting argument.
Since there are $n^2$ pairs of nodes in $[n]^2$ that can form edges in a graph in $\bigraphs$, the size of $\bigraphs$ is $2^{n^2}$.
Let $a$ be the size of the alphabet.
The number of labelings of $2n$ nodes with strings of length $s$ is at most $a^{2ns}$.
In particular, labelings of length 
$s = n/(3\log a)$ can generate
no more than $a^{2n^2/(3\log a)} = 2^{2n^2/3}$ bipartite graphs, which is in $o(2^{n^2})$.
Consequently, almost all graphs in $\bigraphs$ have \readability $\Omega(s) = \Omega(n/\log a) = \Omega(n)$.
The proof for $\digraphs$ is analogous and is omitted.
The proof for variable sized alphabets is given in the \apdx.
\end{proof}

\subsection{Distinctness and a graph family with \readability $\Omega(n)$ }\label{sec:hadamard-graphs}
In this subsection, we will give a technique for proving lower bounds and use it to show a family of graphs with \readability $\Omega(n)$.
For any two vertices $u$ and $v$, the {\em distinctness} of $u$ and $v$ is defined as $DT(u,v) = \max\{|N(u) \setminus N(v)|, |N(v) \setminus N(u)|\}.$
The {\em distinctness} of a bipartite graph $G$, denoted by $DT(G)$,
is defined as the minimum distinctness of any pair of vertices that belong to the same part of the bipartition.
The following lemma relates the distinctness and the \readability of graphs that are not matchings
(for a matching, the \readability is 1, provided that it has at least one edge, and 0 otherwise).
\begin{lemma}\label{cor:lb_distinctness}
For every bipartite graph $G$ that is not a matching, $\readabilityParameter(G)\ge DT(G)+1$.
\end{lemma}
\begin{proof}
By Theorem~\ref{thm:characterization},
it suffices to show that $DT(G)\le \ind(G)-1$.
Let $h = \ind(G)$, let $w:E(G)\to [h]$ be a minimum decomposition of $G$ satisfying the HUB-rule, and
consider the graphs $G_i = G_i^w$, for $i\in [h]$. We need to show that $DT(G)\le h-1$.
Suppose first that each $G_i$ is a matching. Then, since $w$ is a decomposition of $G$,
we have $\Delta(G)\le h$. Moreover, since $G$ is not a matching, it has a pair of distinct vertices, say $u$ and $v$, with a common neighbor,
which implies $DT(G)\le DT(u,v)\le \Delta(G)-1\le h-1$.

Suppose now that there exists an index $j\in [h]$ such that $G_j$ is not a matching, and let $j$ be the maximum such index.
Then, there exist two distinct vertices in $G$, say $u$ and $v$, that have a common neighbor in $G_j$, and therefore
belong to the same biclique of $G_j$.
It follows that $u$ and $v$ are non-isolated twins in $G_j$. Since
$w$ is satisfies the HUB-rule, this implies that
$u$ and $v$ are twins in each $G_i$ with $i\in [j-1]$.
Consequently, for each vertex $x$ in $G$ adjacent to $u$ but not to $v$, the unique $G_i$ with
$(u,x)\in E(G_i)$ satisfies $i>j$. By the choice of $j$,
each such $G_i$ is a matching, and hence there can be at most $h-j$ such vertices $x$. Thus $|N(u)\setminus N(v)|\le h-j$ and similarly
$|N(v)\setminus N(u)|\le h-j$, which implies the desired inequality $DT(G)\le DT(u,v)\le h-j\le h-1$.
\end{proof}

While the distinctness is a much simpler graph parameter than the \Index number, simplicity comes with a price. Namely, the distinctness does not share the nice feature of the \Index number, that of being bounded on exactly the same sets of graphs as the \readability. In Section~\ref{sec:treelb}, we show the existence of graphs (specifically, trees) of distinctness $1$
and of arbitrary large \readability.

We now introduce a family of graphs, inspired by the Hadamard error correcting code, and apply \Cref{cor:lb_distinctness}
to show that their \readability is at least linear in the number of nodes.
We define
$H_k$ as the bipartite graph with vertex sets $V_s = \{v_s\mid v\in \{0,1\}^k \setminus \{0^k\}\}$
and $V_p = \{v_p\mid v\in \{0,1\}^k \setminus \{0^k\}\}$
 and edge set
 $$E(H_k) = \Big\{(v_s,v_p) \in V_s \times V_p ~~{\mid}~~ \sum_{i = 1}^k v_s[i]v_p[i] \equiv 1\pmod 2\Big\}\,.$$
In other words, each vertex has a non-zero $k$-bit codeword vector associated with it and two vertices are adjacent
if the inner product of their codewords is odd.
Let $n = 2^k$. Graph
$H_k$ has $2(n-1)$ vertices, all of degree $n/2$, and thus $(n-1)n/2$ edges.
\Cref{fig:hadamard} illustrates $H_3$.

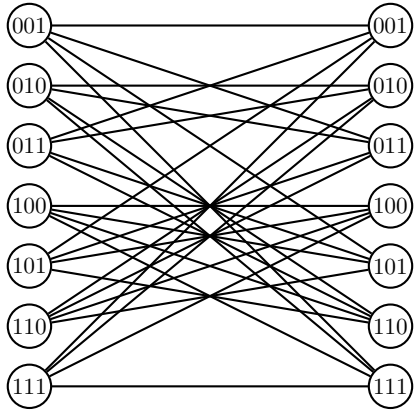
\begin{figure}[t]
        \centering
                \begin{tikzpicture}[-,shorten >=0pt,auto,node distance=6cm,
                    thick, scale=0.8, every node/.style={scale=0.8}]
            \newcommand\hyshift{5cm}
            \tikzstyle{main node}=[circle,draw, inner sep=1pt, minimum width=5pt]
            \node[main node] (l001) {001};
            \node[main node] (r001) [right of=l001] {001};
            \node[main node] (l010) [below of=l001, yshift=\hyshift] {010};
            \node[main node] (r010) [right of=l010] {010};
            \node[main node] (l011) [below of=l010, yshift=\hyshift] {011};
            \node[main node] (r011) [right of=l011] {011};
            \node[main node] (l100) [below of=l011, yshift=\hyshift] {100};
            \node[main node] (r100) [right of=l100] {100};
            \node[main node] (l101) [below of=l100, yshift=\hyshift] {101};
            \node[main node] (r101) [right of=l101] {101};
            \node[main node] (l110) [below of=l101, yshift=\hyshift] {110};
            \node[main node] (r110) [right of=l110] {110};
            \node[main node] (l111) [below of=l110, yshift=\hyshift] {111};
            \node[main node] (r111) [right of=l111] {111};

            \path[]
            (l001) edge (r001)
            (l001) edge (r011)
            (l001) edge (r101)
            (l001) edge (r111)
            (l010) edge (r010)
            (l010) edge (r110)
            (l010) edge (r011)
            (l010) edge (r111)
            (l011) edge (r001)
            (l011) edge (r010)
            (l011) edge (r101)
            (l011) edge (r110)
            (l100) edge (r100)
            (l100) edge (r101)
            (l100) edge (r110)
            (l100) edge (r111)
            (l101) edge (r001)
            (l101) edge (r011)
            (l101) edge (r110)
            (l101) edge (r100)
            (l110) edge (r010)
            (l110) edge (r011)
            (l110) edge (r100)
            (l110) edge (r101)
            (l111) edge (r001)
            (l111) edge (r100)
            (l111) edge (r010)
            (l111) edge (r111)
            ;
        \end{tikzpicture}
     \caption{The graph $H_3$.
            The strings on the vertices correspond to the $k$-bit codeword vectors. }
        \label{fig:hadamard}
\end{figure}
In the \apdx,
we show that every pair of vertices in the same part of the bipartition of $H_k$ has exactly $n/4$ common neighbors. This implies that the distinctness of $H_k$ is $n/4$. Combining this with
\Cref{cor:lb_distinctness}, we obtain
the following theorem.
\begin{restatable}{theorem}{restatablehadamard}
	\label{thm:hadamard}
	$\readabilityParameter(H_k) \geq n/4$+1.
\end{restatable}
This lower bound also translates to directed graphs: applying \Cref{thm:transformation},
there exists digraphs of \readability $\Omega(n)$.
A major open question is: Do there exist graphs that have exponential \readability?
We conjecture that they do, and that the graph family $H_k$ has exponential \readability.
However, since distinctness is $O(n)$,
we note that \Cref{cor:lb_distinctness} is insufficient for proving
stronger than $\Omega(n)$ lower bounds on the \readability.

\subsection{Trees}\label{sec:treelb}
\begin{sloppypar}
	The purely graph theoretic characterization of \readability given by \Cref{thm:tree}
	allows us to derive a sharp upper bound on the \readability of trees.
	Recall that the {\it eccentricity} of a vertex $u$ in a connected graph $G$ is defined as
    ${\it ecc}_G(u) = \max_{v\in V(G)}{\it dist}_G(u,v)$, where ${\it dist}_G(u,v)$ is the number of edges in a shortest path from $u$ to $v$.
	The {\it radius} of a graph $G$ is defined as the minimum eccentricity of a vertex in $G$, that is
	${\it radius}(G) = \min_{u\in V(G)}\max_{v\in V(G)}{\it dist}_G(u,v)$.
\end{sloppypar}

\begin{restatable}{theorem}{restatabletreeradius}
	\label{thm:tree-radius}
For every tree $T$, $\readabilityParameter(T)\le {\it radius}(T)$, and this bound is sharp. More precisely, for every $k\ge 0$ there exists a tree $T$ such that
$\readabilityParameter(T) = {\it radius}(T) = k$.
\end{restatable}
\begin{proof}
Let $T$ be a tree. If $T = K_1$ (the one-vertex tree), then ${\it radius}(T)= \readabilityParameter(T) = 0$ (note that assigning the empty string to the unique vertex of $v$ results in an overlap labeling of $T$). Now, let $T$ be of radius $r\ge 1$ and let $v\in V(T) $ be a vertex of $T$ of minimum eccentricity (that is,
${\it ecc}_T(v) = r$). Consider the distance levels of $T$ from $v$, that is, $V_i = \{w\in V(T)\mid {\it dist}_T(v,w) = i\}$ for $i\in \{0,1,\ldots, r\}$.
Also, for all $i \in [r]$, let $E_i$ be the set of edges in $T$ connecting a vertex in $V_{i-1}$ with a vertex in $V_{i}$.
Then $\{E_1,\ldots, E_r\}$ is a partition of $E(T)$ and the decomposition $w:E(T)\to [r]$ given by $w(e) = i$ if and only if $e\in E_i$
is well defined. We claim that $w$ satisfies the $P_4$-rule.
Let $P=(v_1,v_2,v_3,v_4)$ be an induced $P_4$ in $T$, and let $i = w(v_1,v_2)$, $j = w(v_2,v_3)$, $k = w(v_3,v_4)$.
Suppose that $j = \max\{i,j,k\}$. We may assume without loss of generality that $v_2\in V_{j-1}$ and $v_3\in V_j$.
Since $T$ is a tree, $v_2$ is the only neighbor of $v_3$ in $V_{j-1}$, which implies that $v_4\in V_{j+1}$ and consequently $k = j+1$, contrary to the assumption
$j = \max\{i,j,k\}$. Thus, the $P_4$-rule is trivially satisfied for $w$. By Theorem~\ref{thm:tree}, we have
$\readabilityParameter(T)\le \max_{e\in E(T)}w(e) = r = {\it radius}(T)$.

To show that for every $k\ge 0$ there exists a tree $T$ with $\readabilityParameter(T) = {\it radius}(T) = k$, we proceed by induction.
We will construct a sequence $\{(T_i,v_i)\}_{i\ge 0}$ where $T_i$ is a tree, $v_i$ is a vertex in $T_i$ with ${\it ecc}_{T_i}(v_i)\le i$,
the degree of $v_i$ in $T_i$  is $i$, and
$\readabilityParameter(T_i) = {\it radius}(T_i) = i$.
For $i = 0$, take $(T_0, v_0)  = (K_1, v_0)$ where $v_0$ is the unique vertex of $K_1$. This clearly has the desired properties.
For $i\ge 1$, take $i$ disjoint copies of $(T_{i-1},v_{i-1})$, say $(T_{i-1}^j,v_{i-1}^j)$ for $j \in [i]$, add a new vertex $v_{i}$, and join $v_{i}$ by an edge to each $v_{i-1}^j$ for $j \in [i]$. Let $T_i$ be the so constructed tree.
Clearly, the degree of $v_i$ in $T_i$  is $i$, and ${\it ecc}_{T_i}(v_i)\le 1+
{\it ecc}_{T_i}(v_{i-1}) \le 1+(i-1) = i$, which implies that ${\it radius}(T_i)\le i$.
On the other hand, we will show that $\readabilityParameter(T_i)\ge i$, which together with inequality $\readabilityParameter(T_i)\le {\it radius}(T_i)$ will imply the desired conclusion ${\it radius}(T_i) = \readabilityParameter(T_i) = i$. Suppose for a contradiction that $\readabilityParameter(T_i)<i$. Then,
by Lemma~\ref{lem:fourstrings}, there exists a decomposition $w$ of $T_i$ of size $i-1$
satisfying the $P_4$-rule. In particular, this implies $i\ge 2$. Since the degree of $v_i$ in $T_i$ is $i$, there exist two edges incident with $v_i$, say
$(v_i,v_{i-1}^j)$ and $(v_i,v_{i-1}^k)$ for some $j\neq k$ such that
$w(v_i,v_{i-1}^j) = w(v_i,v_{i-1}^k)$. Let $w_1$ denote this common value. Let $x$ be a neighbor of $v_{i-1}^j$ in $T_{i-1}^j$.
(Note that $x$ exists since $v_{i-1}^j$ is of degree $i-1\ge 1$ in $T_{i-1}^j$.)
Then, $(x,v_{i-1}^j, v_i, v_{i-1}^k)$ is an induced $P_4$ in $T_i$. We claim that
$w(x,v_{i-1}^j)>w_1$. Indeed, if $w(x,v_{i-1}^j)\le w_1$ then we have 
$\max\{w(x,v_{i-1}^j), w(v_{i-1}^j,v_i), w(v_i,v_{i-1}^k)\} = \max\{w(x,v_{i-1}^j), w_1, w_1\} = w_1$, while $w_1\ngeq w_1+w(x,v_{i-1}^j)$, contrary to the $P_4$-rule.
Since $x$ was an arbitrary neighbor of $v_{i-1}^j$ in $T_{i-1}^j$, we infer that
every edge $e$ in  $T_{i-1}^j$ incident with $v_{i-1}^j$ satisfies $w(e) >w_1$.
In particular, this leaves a set of at most $i-2$ different values that can appear on these $i-1$ edges (the value $w_1$ is excluded), and hence again
there must be two edges of the same weight, say $w_2$. Clearly, $w_2>w_1$ and $i>2$.
Proceeding inductively, we construct a sequence of edges $e_1, e_2, \ldots, e_{i}$ forming a path in $T_{i}$ from $v_i$ to a leaf and satisfying
$w_1<w_2<\ldots<w_i$, where $w_i = w(e_i)$. This implies that all the weights $w_1,\ldots, w_i$ are distinct, contrary to the fact that
the range of $w$ is contained in the set $[i-1]$. This contradiction shows that $\readabilityParameter(T_i)\ge i$ and completes the proof.
\end{proof}
Note that for every $k\ge 2$, the tree $T_k$ of radius $k$ constructed in the proof of Theorem~\ref{thm:tree}
has a pair of leaves in the same part of the bipartition and is therefore of distinctness $1$. This shows that the \readability of a graph cannot be upper-bounded by any function of its distinctness (cf.~\Cref{cor:lb_distinctness}).

\section{Conclusion}
In this paper, we define a graph parameter called \readability,
and initiate a study of its asymptotic behavior.
We give purely graph theoretic parameters (i.e., without reference to strings)
that are exactly (respectively, asymptotically) equivalent
to \readability for trees (respectively, $C_4$-free graphs);
however, for general graphs,
the \Index number is equivalent to \readability only in the sense that it is bounded on the same set of graphs.
While an $\ell$-decomposition always satisfies the HUB-rule,
the converse is not true.
For example, a decomposition of $P_4$ with weights $4,5,3$ satisfies the HUB-rule but cannot be
achieved by an overlap labeling (by \Cref{lem:fourstrings}).
For this reason, the upper bound given by \Cref{lem:generalized-BM} leaves a gap with the lower bound of \Cref{lem:hubrule}.
We are able to describe other properties that an $\ell$-decomposition must satisfy
(not included in the paper),
however, we are not able to exploit them to close the gap.
It is a very interesting direction to find other necessary rules
that would lead to a graph theoretic parameter that would more tightly match \readability on general graphs than the \Index number.

Consider $\readabilityParameter(n) = \max \setof{\readabilityParameter(D)}{D \text{ is a digraph on $n$ vertices}}$.
We have shown $\readabilityParameter(n) = \Omega(n)$ and know from \cite{BM02} that $\readabilityParameter(n) = O(2^n)$.
Can this gap be closed? Do there exist graphs with \readability $\Theta(2^n)$ (as we conjecture),
or, for example, is \readability always bounded by a polynomial in $n$?
Questions regarding complexity are also unexplored, e.g., given a digraph,
is it NP-hard to compute its \readability?
For applications to bioinformatics,
the length of reads can be said to be poly-logarithmic in the number of vertices.
It would thus be interesting to further study the structure of graphs that have poly-logarithmic \readability.

\paragraph{Acknowledgements.}
P.M.~and M.M.~would like to thank Marcin Kami\'nski for preliminary discussions.
P.M.~was supported in part by NSF awards DBI-1356529 and CAREER award IIS-1453527.
M.M.~was supported in part by the Slovenian Research Agency (I$0$-$0035$, research program P$1$-$0285$ and research
projects N$1$-$0032$, J$1$-$5433$, J$1$-$6720$, and J$1$-$6743$).
S.R.~was supported in part by NSF CAREER award CCF-0845701, NSF award AF-1422975 and the Hariri Institute for Computing and Computational Science and Engineering at Boston University.

\bibliographystyle{alpha}
\bibliography{readability}

\newpage
\appendix
\section{Appendix: deferred proofs}\label{app:all}
\subsection{\ReadabilityFull~of bipartite graphs and digraphs}
In this subsection, we prove the following theorem.
\restatabletransformation*
Recall that $\bigraphs$ is defined as the set of balanced bipartite graphs with nodes $[n]$ in each part.
To disambiguate the two partitions, we label the vertices of $G = (V_s, V_p, E) \in \bigraphs$
using notation $V_s = \setof{i_s}{i \in [n]}$ and $V_p = \setof{i_p}{i \in [n]}$.

For the proof, we define the following transformation.
Let $D = ([n],A) \in \digraphs$.
Define $\phi(D) = (V_s, V_p, E)$ as the bipartite graph with
$V_s = \setof{i_s}{i \in [n]}$, $V_p = \setof{i_p}{i \in [n]}$.
and $E = \setof{(i_s,j_p)}{(i,j)\in A}$.
This transformation was proposed in \cite{BM02}.
Similarly, we define the transformation $\psi$, as follows.
Given a bipartite graph $G = (V_s, V_p, E) \in \bigraphs$,
we define $\psi(G) = ([n], A)$ where $A = \setof{(i,j)}{(i_s, j_p) \in E}$.
It is easy to see that $\psi$ is a bijection from $\bigraphs$ to $\digraphs$, as required,
and $\phi$ is its inverse.

The following two lemmas prove the \readability bounds stated in the theorem.
\begin{restatable}{lemma}{restatedoublingtransformone}\label{lem:phi}
	Let $D = (V, A) \in \digraphs$ be a digraph with $A\neq\emptyset$.
	Then $\readabilityParameter(\phi(D)) < \readabilityParameter(D)$.
\end{restatable}
\begin{proof}
	Let $\ell$ be an injective overlap labeling of $D$. Since $A\neq \emptyset$, we have $\len(\ell)\ge 1$.
	Define a labeling $\ell_\phi$ of $\phi(D)$ as follows.
	For $w \in V$, let $\ell_\phi(w_s) = \ell(w)[2 .. |\ell(w)|]$ and
	let $\ell_\phi(w_p) = \ell(w)[1 .. |\ell(w)| - 1]$. (If $|\ell(w)| = 1$, then each of $\ell_\phi(w_s)$ and $\ell_\phi(w_p)$ is the empty string.)
	It is clear that $\ell_\phi$ is a labeling of $\phi(D)$ of length $\len(\ell) - 1$.
    We claim that $\ell_\phi$  is an overlap labeling of $\phi(D)$.
    Suppose that $(u_s,v_p)\in E(\phi(D))$. Then $(u,v)\in A$, which implies $\ov_\ell(u,v)>0$.
    Also,  $\ov_\ell(u,v) < \len(\ell)$.
    Consequently, the shortest overlap between
    $\ell(u)$ and $\ell(v)$ yields an overlap between $\ell_\phi(u_s)$ and $\ell_\phi(v_p)$, implying
    $\ov_{\ell_\phi}(u_s,v_p) > 0$.
    Conversely, the condition  $\ov_{\ell_\phi}(u_s, v_p) >0$ implies
    $0 < \ov_\ell(u,v) < \len(\ell)$.
    Therefore,
    $(u,v)\in A$ and, by the definition of $\phi(D)$, also $(u_s,v_p)\in E(\phi(D))$.
    This shows that $\readabilityParameter(\phi(D)) \le \readabilityParameter(D)-1$.
\end{proof}

\begin{sloppypar}
\begin{restatable}{lemma}{restatedoublingtransformtwo}
	Let $G = (V_s, V_p, E) \in \bigraphs$.
	Then \hbox{$\readabilityParameter(\psi(G)) \leq 2 \cdot \readabilityParameter(G)  + 1$}.
\end{restatable}
\end{sloppypar}
\begin{proof}
	Let $\ell_G$ be an overlap labeling of $G$ and
	let $D = (V,A) = \psi(G)$, with $V = [n]$.
	For $w \in V$, define $\ell(w) = \ell_G(w_p) \cdot w \cdot \ell_G(w_s)$.
	Here, $w$ is treated as a character in the alphabet $[n]$.
	We assume without loss of generality that these characters are distinct from the alphabet over which
	$\ell_G$ is defined.
	It is clear that $\ell$ is a labeling of $D$ of length $2\cdot\len(\ell_G) + 1$.
	We claim that $\ell$ is an injective overlap labeling of $D$.
	For every vertex $w\in V$, its label contains a distinct middle character corresponding to $w$,
	which implies injectivity.
	Now, suppose that $(u,v)\in A$. Then $(u_s,v_p)\in E$, which implies $\ov_{\ell_G}(u_s,v_p) > 0$.
    By construction of $\ell$, it follows that
    $0 < \ov_\ell(u,v) \leq \len(\ell_G) < \len(\ell)$.
    Conversely, suppose that $\ov_\ell(u,v) > 0$. By construction of $\ell$,
    it follows that $\ov_\ell(u,v) \le \len(\ell_G)$. Therefore,
    $\ov_{\ell_G}(u_s,v_p) = \ov_\ell(u,v) > 0$, which implies
    $(u_s,v_p)\in E$ and consequently
    $(u,v)\in A$.
    This shows that $\readabilityParameter(\psi(G)) \le 2\cdot \readabilityParameter(G)+1$.
\end{proof}
Given $G \in \bigraphs$, we can apply the two lemmas to derive the inequality of \Cref{thm:transformation}:
$$
\readabilityParameter(G) = \readabilityParameter(\phi(\psi(G)) < \readabilityParameter(\psi(G)) \leq 2\cdot \readabilityParameter(G) + 1.
$$

\subsection{Trees and $C_4$-free graphs}
\restatablefourstrings*
\begin{proof}
Let $G = (V_s, V_p, E)$.
Denote by $w$ be the $\ell$-decomposition.
Suppose for the sake of contradiction that $w$ violates the
$P_4$-rule. Then, there exists an induced four-vertex path $P = (u_1,u_2,u_3,u_4)$ in $G$
with $u_1\in V_p$ (and consequently $u_2,u_4\in V_s$ and $u_3\in V_p$)
 such that $\max\{w(v_1, v_2), w(v_2, v_3), w(v_3, v_4)\} = w(v_2, v_3) <w(v_1, v_2) + w(v_3, v_4)$.
Then, $b=\max\{a,b,c\}$ and $b < a + c$, where
$a = \ov_\ell(u_2,u_1)$, $b = \ov_\ell(u_2,u_3)$, and $c = \ov_\ell(u_4,u_3)$.
We will show that there exists an overlap from $\ell(u_1)$ to $\ell(u_4)$ of length $a + c - b$,
which will prove the lemma, by contradicting the fact that $\ell$ is an overlap labeling and $(u_4,u_1)\not\in E$ (as $P$ is an induced $P_4$).

Let $r$ be the length of $\ell$. Writing the overlaps in terms of substrings, we obtain that
    $\suf_a(\ell(u_2)) = \pre_a(\ell(u_1))$,
    $\suf_b(\ell(u_2)) = \pre_b(\ell(u_3))$, and
    $\suf_c(\ell(u_4)) = \pre_c(\ell(u_3))$.
	Let $d = a + c - b$.
	Note that $1 \leq d \leq \min\{a,c\}$.
	Applying the equalities, we get
    $\pre_d(\ell(u_1)) =
	\ell(u_2)[r - a + 1 .. r - a + d] =
	\ell(u_3)[c - d + 1 .. c] =
    \suf_d(\ell(u_4))$, establishing the existence of the desired overlap.
\end{proof}

\restatablepathexists*
\begin{proof}
	We prove the statement by induction on $m\in \{0,1,\ldots, |E|\}$.
	For the base case, $\ell_0$ does not label any positions.
	Now, assume that $\ell_m$ satisfies the lemma and
	consider the new positions labeled by $\ell_{m+1}$, with $e_{m+1} = (u,v)$.
	Recall that $A$ is a possibly empty string of new characters inserted into the middle of the new labels.
	A position of $u$ labeled with a character from $A$ is adjacent to the position of $v$ labeled with the same character,
	and since the characters are new, these are the only two positions labeled with this character.
	Now, each new position of $u$ that is not labeled with a character from $A$ is labeled with a character from $\ell_m(v)$.
    By the induction hypothesis, $v$ is connected by a path to all vertices with occurrences of the same character in $G^m$, which implies the
 same statement for $u$ in $G^{m+1}$ (using the fact that $E(G^{m+1}) = E(G^{m})\cup \{e_{m+1}\}$).
	The case of the new characters in the label of $v$ is symmetric.
\end{proof}

\subsection{General graphs}
\restatablehubrule*
\begin{proof}
	Denote the vertices and edges of the graph as usual: $G=(V_s, V_p, E)$.
	Consider the $\ell$-decomposition.
	Fix $i\in[k]$.
	First, we show that $G_i$ is a union of disjoint bicliques.	
Observe that a bipartite graph is a disjoint union of bicliques if and only if it contains no induced $P_4$,
where a $P_4$ denotes the path on $4$ vertices and $3$ edges.
Therefore, it suffices to prove that $G_i$ does not contain any induced $P_4$.
Consider a $4$-vertex path $(u,x,y,z)$ in $G_i$. We will show that $G_i$ contains the edge $(u,z)$.
Since each edge of the path is
	in $G_i$, the corresponding overlaps imply that $inner_i(\ell(u))=inner_i(\ell(x))=inner_i(\ell(y)) =
	inner_i(\ell(z))$. Thus, $inner_i(\ell(u)) = inner_i(\ell(z))$.
 To complete the proof that $(u,z)\in E(G_i)$, it remains to show that $(u,z)\notin E(G_j)$ for all $j \in[ i-1]$. For the sake of
	contradiction suppose $inner_j(\ell(u)) = inner_j(\ell(z))$ for some $j\in[ i-1]$. Then
	$inner_j(\ell(u)) = inner_j(\ell(x))$ and, consequently, $(u,x)$ is in $E(G_j)$, which contradicts that
	it is in $E(G_i)$. Therefore, $(u, z)\in E(G_i)$. This completes the proof that $G_i$ is a disjoint union of bicliques.

Next we show that the $\ell$-decomposition is hierarchical. i.e. satisfies the second condition of the HUB-rule definition.
Fix $i\in\{2,\dots,k\}$ and consider two non-isolated twins $u,v$ in $G_i$.
By definition of non-isolated twins, there is a vertex $z$ that is adjacent to both $u$ and $v$ in $G_i$.
By definition of $G_i$, we get $inner_i(\ell(u))  = inner_i(\ell(z))= inner_i(\ell(v)).$
Therefore, for all $j\in[i-1]$, the corresponding inner affixes of labels of $u$ and $v$ are the same:
$inner_j(\ell(u))  = inner_j(\ell(v))$. Consequently, in $G_j$, every neighbor of $u$ must be a neighbor of $v$, and vice versa. That is, $u$ and $v$ are twins in $G_j$ for all $j\in[i-1]$, completing the proof of the lemma.
\end{proof}

\restategeneralizedBM*
\begin{proof}
	In \Cref{sec:hub-up}, we described how to inductively construct a labeling $b$ of the appropriate length.
	It remains to prove that the final labeling is an overlap labeling of $G$. It is easy to see that the initial labeling of $G_1$ is an overlap labeling. Now we show that if $t$ is an overlap labeling of $G_1\cup\dots\cup G_i$, our construction yields an overlap labeling of $G_1\cup\dots\cup G_{i+1}$.

Suppose first that $(u,v)$ is an edge of $G_1\cup\dots\cup G_{i+1}$.
If $(u,v)$ is an edge of $G_{i+1}$ then, by construction, the labels of $u$ and $v$ after $i+1$ steps are identical,
and consequently they overlap.
If $(u,v)$ is not an edge of $G_{i+1}$, then it is an edge of $G_1\cup\dots\cup G_{i}$, and the bicliques $B$ and $B'$ of $G_{i+1}$ containing $u$ and $v$, respectively, are distinct. This implies that the labels of $u$ and $v$ after $i+1$ steps are of the form
$x\cdot a_B\cdot t(u)$ and $t(v)\cdot a_{B'}\cdot y$, respectively, for some (possibly empty) strings $x,y,a_B,$ and $a_{B'}$, where
$t(u)$ and $t(v)$
are the respective labels of $u$ and $v$ after $i$ steps.
Since, by the induction hypothesis, $t(u)$ and $t(v)$ overlap, so do the extended labels.

Finally, if $(u,v)\in V_s\times V_p$ is a pair of nonadjacent vertices of $G_1\cup\dots\cup G_{i+1}$,
then $u$ and $v$ are nonadjacent in $G_1\cup\dots\cup G_{i}$. By induction hypothesis, their labels after $i$ steps, $t(u)$ and $t(v)$, do not overlap.
Since $u$ and $v$ are also not adjacent in $G_{i+1}$, the bicliques of $G_{i+1}$ containing $u$ and $v$, say $B$ and $B'$, are distinct, and
thus the labels of $u$ and $v$ after $i+1$ steps are
of the form $x\cdot a_B\cdot t(u)$ and $t(v)\cdot a_{B'}\cdot y$, respectively. Moreover, if both $x\cdot a_B$ and $a_{B'}\cdot y$ are nonempty then $a_B \neq a_{B'}$. Hence,
by construction, the two labels do not overlap.
This completes the proof.
\end{proof}

\subsection{Almost all graphs have \readability $\Omega(n/\log n)$}
In this subsection, we give a proof of the following theorem.
\restatablebinary*
We will need the following reduction, implicitly shown in \cite{BM02}.
\begin{property}[\cite{BM02}]\label{prop:alphabet}
	Let $G$ be a digraph or a bipartite graph, let
	$\Sigma$ and $\Sigma'$ be alphabets with $|\Sigma|\ge |\Sigma'|\ge 2$,
	and let $\ell$ be an overlap labeling of $G$ over $\Sigma$.
	Then there exists an overlap labeling $\ell'$ of $G$ over $\Sigma'$ such that
	$\len(\ell') \le (2\log_{|\Sigma'|}|\Sigma|+1)\cdot \len(\ell)$.
\end{property}
The proof of \Cref{lem:binary} for constant sized alphabets is in the main text.
For variable sized alphabets, we give the proof here.
\begin{proof}[Proof (variable sized alphabets)]
The proof of the constant sized alphabet shows that only $o(2^{n^2})$ graphs in
$\bigraphs$ have \readability at most $n/3$ over the binary alphabet.
It therefore suffices to show that every graph in $\bigraphs$ of \readability at most
$n/(15\log_2 n)$ (over an unrestricted alphabet) has \readability at most $n/3$ over the binary alphabet.
This is indeed the case. Suppose that $G\in \bigraphs$
is of \readability $r \le n/(15\log_2 n)$, and fix
an overlap labeling $\ell$ of $G$ of length $r$. Since $\ell$ uses $2nr$ characters in total,
the alphabet size of labeling $\ell$ can be assumed to be at most $2nr$.
By Property~\ref{prop:alphabet}, $G$ has an overlap labeling $\ell'$ over the binary alphabet such that
 $\len(\ell') \le (2\log_2(2nr)+1)r$.
Since $2nr\le n^2$, we have $2\log_2(2nr)+1\le 5\log_2 n$ and consequently
the \readability of $G$ over the binary alphabet is at most
$\len(\ell')\le 5r\log_2 n \le n/3$.
The proof for $\digraphs$ is analogous and is omitted.
\end{proof}

\subsection{Graph family with \readability $\Omega(n)$}\label{app:Hadamard}
We prove the following lemma, which was used in Section~\ref{sec:hadamard-graphs} to prove \Cref{thm:hadamard}.

\begin{lemma}\label{lem:Hadamard-neighbors}
In graph $H_k$,
if $i$ vertices have a common neighbor, then they have
at least $2^{k-i} = n/2^i$ common neighbors.
Moreover, if two vertices have a common neighbor, then they have
exactly $n/4$ common neighbors.
\end{lemma}
\begin{proof}
Suppose that vertices $w_1,\ldots, w_i\in \{0,1\}^{k}\setminus\{0^k\}$ in the same part of the bipartition of $H_k$ have a common neighbor.
Then the set $X$ of all vectors $x\in \{0,1\}^k$
such that $w_j^\top x = \sum_{p = 1}^kw_j[p]x[p]  \equiv 1\pmod {2}$ is non-empty.
Notice that $X\subseteq \{0,1\}^k$ is the set of solutions of the equation $W x  = \mathbf{1}$
over the field $GF(2)$, where $W$ is the $i\times k$ matrix with the rows  formed by the $w_j$'s, and
$\mathbf{1}$ is the all-one vector of length $i$.
The set $X$ forms an affine subspace of the vector space $\{0,1\}^k$ over $GF(2)$
of dimension $k-r$, where $r = {\it rank}(W)$.
Therefore, vertices $w_1,\ldots, w_i$ have exactly $|X| = 2^{k-r}$ common neighbors.
Since $r\le i$, we obtain $|X|\ge 2^{k-i}$.

If $i = 2$, then the rank of $W$ is exactly $2$, which implies the second part of the lemma.
\end{proof}

\end{document}